\documentclass[11pt]{article}
\usepackage{latexsym}
\usepackage{theorem}
\usepackage{graphicx}
\usepackage{amsmath,color}
\usepackage{amsfonts}
\usepackage{natbib}
\usepackage{soul}

\theorembodyfont{\rmfamily}
\newtheorem{theorem}{Theorem}

\newtheorem{lemma}[theorem]{Lemma}
\newtheorem{proposition}[theorem]{Proposition}

\theoremstyle{break}

\usepackage{mathtools}
\DeclarePairedDelimiter\ceil{\lceil}{\rceil}
\DeclarePairedDelimiter\floor{\lfloor}{\rfloor}

\newenvironment{proof}{\paragraph{Proof.}}{\hfill$\square$}

\title{The Search and Rescue Game on a Cycle}

\date{}
\author{Thomas Lidbetter\thanks{Department of Systems \& Information Engineering, University of Virginia, VA 22903, USA, tlidbetter@virginia.edu (corresponding author)}  \thanks{Rutgers Business School, 1 Washington Park, Newark, NJ 07102, USA, tlidbetter@business.rutgers.edu} \and Yifan Xie\thanks{Department of Industrial and Systems Engineering, Rutgers University, 96 Frelinghuysen Rd, Piscataway, NJ 08854, USA, yifan.xie@rutgers.edu}}

\providecommand{\keywords}[1]{\textbf{\textbf{Keywords:}} #1}

\begin{document}
	
\maketitle

\begin{abstract}
\noindent We consider a search and rescue game introduced recently by the first author. An immobile target or targets (for example, injured hikers) are hidden on a graph. The terrain is assumed to be dangerous, so that when any given vertex of the graph is searched, there is a certain probability that the search will come to an end, otherwise with the complementary {\em success probability} the search can continue. A Searcher searches the graph with the aim of finding all the targets with maximum probability. Here, we focus on the game in the case that the graph is a cycle. In the case that there is only one target, we solve the game for equal success probabilities, and for a class of games with unequal success probabilities. For multiple targets and equal success probabilities, we give a solution for an adaptive Searcher and a solution in a special case for a non-adaptive Searcher. We also consider a continuous version of the model, giving a full solution for an adaptive Searcher and approximately optimal solutions in the non-adaptive case.
\end{abstract}


\newpage

\section{Introduction}

The {\em search and rescue game} was introduced in \cite{Lidbetter20} to model a scenario in which a target or targets hidden on a graph must be located by a Searcher who faces some sort of danger in the searching process. For instance, in a search and rescue operation taking place in dangerous terrain, the Searcher could become incapacitated; when searching in a military context, the Searcher could be captured by an opponent. 

More specifically, the model assumes that target or targets are hidden on a graph by an adversary, or Hider, and the Searcher searches the vertices of the graph one-by-one until finding all the targets. When searching each vertex, there is independently some given probability that the search will be cut short, otherwise, with the complementary {\em success probability}, the search can continue. The mode of search considered is known as {\em expanding search}, independently introduced by \cite{AP12} and \cite{AL13}. An expanding search of a graph, starting at a given node, chooses subsequent nodes in such a way that each node chosen is adjacent to some previously chosen node. This search paradigm is appropriate for situations in which the cost of retracing one's steps is negligible. For example, when sweeping an area for landmines, once an area has been found to be safe, it can be traversed quickly compared to the slower pace required to traverse a dangerous area. Expanding search can also be applied to a setting in which a large team of searchers successively splits into smaller and smaller groups (see \cite{AL13} for more details).

The model uses the framework of zero-sum games. The game is between a Searcher who aims to maximize the probability of finding the targets, and a Hider who aims to minimize this probability. We seek optimal mixed (randomized) strategies for both players and the value of the game.

\cite{Lidbetter20} solved the search and rescue game in two settings. In the first setting, there is an arbitrary, known number of targets but no graph structure; in the second setting, the problem was solved for the case of only one target hidden on a tree. In this paper, we consider the game in the case of multiple targets hidden on cycle graphs, which could be considered the simplest graphs that are not trees. An example to have in mind could be a search and rescue operation on a hiking trail that forms a loop.

This work lies in the general area of search games. Good overviews on the topic are \cite{AG03} and \cite{H16}. On the topic of search and rescue, \cite{LBA22} studied a rendezvous problem in which two mobile agents wish to find each other. \cite{A11} considered a find-and-fetch problem which can be considered as a search and rescue game. \cite{BK17} considered a search game in which a Hider is located on a cycle. 

Much of this paper concentrates on problems where multiple targets have to be found. There has not been much work on search games with multiple hidden targets. \cite{Lidbetter13} considered a Searcher who wishes to find a cost-minimizing search that locates multiple hidden targets. \cite{N90} considered a game in which two Searchers each try to find a different target before the other. \cite{S87} and \cite{AZ87} both consider search problems in which a Searcher wishes to find one of many hidden targets, but these papers are not game theoretic.

This paper is arranged as follows. We start in Section~\ref{sec:prelim} by defining the game precisely. In Section~\ref{sec:one-target} we consider the game played on a cycle in the case of only one target. In the case of equal success probabilities, the game has a simple solution, but for non-equal success probabilities, the solution appears to be non-trivial. We give a sufficient condition for the game to have a particularly simple solution, and we also give approximately optimal strategies for both players, which perform well if the success probabilities are not too ``spread out''.

In Section~\ref{sec:multiple}, we turn to the case of multiple targets, considering both the settings of {\em adaptive} and {\em non-adaptive} search. Adaptive search is characterized by the Searcher's freedom to choose the next node of search based on the history of the search so far; in non-adaptive search, the Searcher must set out from the beginning the order of search of the nodes. We give a full solution of the adaptive version of the game for equal success probabilities. The non-adaptive version appears harder to analyze and we give the solution in the simple case of three nodes.

Finally, in Section~\ref{sec:cont}, we consider a continuous version of the game, in which the Hider can hide the targets continuously on a cycle network, viewed as a continuous measure space. We use a continuous version of expanding search as defined in \cite{AL13}, where the area that the Searcher has searched is a connected subset of the space that increases at a constant rate.  For the adaptive case, we give a full solution for an arbitrary number of targets. For the non-adaptive case, we give strategies that are close to being optimal, in the sense that they give upper and lower bounds on the value that are very close to each other. We conclude in Section~\ref{sec:conc}.

\section{Preliminaries}
\label{sec:prelim}

In this section we define the search and rescue game precisely, starting with the version of the game with no graph structure, then going on to the ``graph version'' of the game. We also note a result from~\cite{Lidbetter20} which will be helpful later. 

The search and rescue game is a zero-sum game played between a Hider and a Searcher, where the Hider distributes $k$ targets among a finite set $S$ of hiding places and the Searcher aims to maximize the probability of finding them all. A pure strategy for the Hider is a subset $H \in S^{(k)}$, where $S^{(k)}$ denotes all subsets of $S$ of cardinality $k$. The set $H$ corresponds to the $k$ hiding places. A pure strategy for the Searcher is a permutation of $S$, specifying the order of search. More precisely, a pure strategy is a bijection $\sigma:\{1,\ldots,|S|\}\rightarrow S$, where $\sigma(j)$ is interpreted as the $j$th vertex to be searched. 

To each location $i \in S$, a probability $p_i$ is associated, where $0<p_i<1$. This is the probability that the Searcher is not captured herself when searching location $i$, and we refer to $p_i$ as the {\em success probability} of location $i$. The payoff of the game is the probability that the Searcher rescues all $k$ targets without being captured herself. More precisely, for a given pure Hider strategy $H$ and a given pure Searcher strategy $\sigma$, let $\sigma^{-1}(H)$ denote the positions under $\sigma$ of the elements of $H$. Then the payoff of the game is
\[
P(H,\sigma) \equiv \prod_{\{i: i \le \max \sigma^{-1}(H)\}} p_{\sigma(i)}.
\]
That is, $P(H,\sigma)$ is the product of the success probabilities of all the vertices searched up to and including the last vertex that is a member of $H$. A mixed strategy $s$ for the Searcher is a probability distribution over the set of permutations and a mixed strategy $h$ for the Hider is a probability distribution over the set $S^{(k)}$ of subsets of locations of cardinality $k$. For given mixed strategies $s$ and $h$ we will sometimes use the notation $P(h,s)$ to denote the expected payoff.

Since this is a finite zero-sum game, the Searcher has optimal max-min strategies, the Hider has optimal min-max strategies and the game has a value $V$ given by
\[
V = \max_s \min_H P(s,H) = \min_h \max_\sigma P(\sigma,h).
\]
\cite{Lidbetter20} solved the game, and we restate the solution here since we will make use of it later.

\begin{theorem}[Theorem 3 of \cite{Lidbetter20}] \label{thm:Lidbetter}
In the search and rescue game, it is optimal for the Hider to choose each subset $A \in S^{(k)}$ with probability
\[
q_A \equiv \lambda_k \prod_{i \in A}\frac{1-p_i}{p_i}, \text{ where } \lambda_k = \left( \sum_{B \in S^{(k)}} \prod_{i \in B}\frac{1-p_i}{p_i} \right)^{-1}.
\]
It is optimal for the Searcher to choose a subset $A \in S^{(k)}$ of locations to search first with probability $q_A$, then search the remaining elements of $S$ in a uniformly random order.  

If $k=1$, the value $V$ of the game is given by
\[
V = \frac{1 - \prod_{i \in S} p_i }{\sum_{i \in S} (1-p_i)/p_i}.
\]
\end{theorem}

\cite{Lidbetter20} also considered an extension of the game in which a graph structure is imposed, so that the locations $S$ are vertices of a graph $G$. The Searcher must choose an ordering of the vertices which corresponds to an {\em expanding search} of $G$, as defined in \cite{AL13}. An expanding search is a sequence $\sigma$ of the vertices in $S$ starting with some {\em root vertex} such that for each $j=2,\ldots,|S|$ the vertex $\sigma(j)$ in the $j$th place must be a neighbor of some other previously chosen vertex. That is, $\sigma(1)=O$ and each $\sigma(j)$ is a neighbor of one of the vertices in $\{\sigma(1),\sigma(2),\ldots,\sigma(j-1)\}$ for $j>1$. This extension of the game was solved in \cite{Lidbetter20} for $k=1$ in the case that the graph is a tree. We denote the search and rescue game played on a graph $G$ by $\Gamma=\Gamma(G)$

In this paper, we consider the game played on a cycle $C_n$, which we define as the graph with vertices $\{0,1,2,\ldots,n\}$ and edges $\{j,j+1\}$ for $j=0,\ldots,n-1$ and the edge $\{n,0\}$. Note that $C_n$ has $n+1$ vertices (contrary to the convention). Vertex 0 is the root vertex, and we assume that $p_0=1$, since any expanding search necessarily starts with vertex $0$. We may also assume that the Hider does not hide any targets at vertex $0$.

Note that Theorem~\ref{thm:Lidbetter} gives an upper bound on the value of the game, since the Hider strategy described in the theorem is available to use on any graph. In general the Searcher will not have a strategy that can meet this bound. We summarize this observation in the lemma below.

\begin{lemma} \label{lem:upperbound}
The value of the search and rescue game $\Gamma(C_n)$ is bounded by the value given in Theorem~\ref{thm:Lidbetter}.
\end{lemma}

%We will denote a mixed strategy for the Hider by a lower case letter $h$. Formally, $h$ is a set of probabilities $\{x_H:H \in [n]^{(k)}\}$ that sum to 1. For the Searcher, a mixed strategy will be denoted by a lower case letter $s$, which is a set of probabilities $y_\sigma$ summing to 1 such that $\sigma$ is a pure Searcher strategy.

\section{Searching for One Target}
\label{sec:one-target}

In this section, we consider the game $\Gamma(C_n)$ in the case that there is only $k=1$ target. In this case, a pure strategy for the Hider is simply an element $j\in C_n$. For a given pure strategy $\sigma$ of the Searcher, the payoff is given by 
\[
P(j,\sigma) \equiv \prod_{\{i: i \le \sigma^{-1}(j)\}} p_{\sigma(i)}.
\]

\subsection{Equal detection probabilities}

We begin by considering the case with equal detection probabilities. In this case the game has a simple solution in which both players mix between only two pure strategies. We denote the Searcher strategy $\sigma^C \equiv (1,2,\ldots,n)$ of traversing the whole cycle clockwise by $\sigma^C$ and the strategy $\sigma^A \equiv (n,n-1,\ldots,1)$ of traversing the whole cycle anticlockwise by $\sigma^A$.

\begin{theorem} \label{thm:equal-p}
Suppose $p_1=p_2=\cdots=p_n=p$. The value of the game is $(p^{\floor*{(n+1)/2}} + p^{\ceil*{(n+1)/2}})/2$. It is optimal for the Hider to choose vertices $\floor*{(n+1)/2}$ or $\ceil*{(n+1)/2}$ with equal probability. It is optimal for the Searcher to choose $\sigma^C$ or $\sigma^A$ with equal probability.
\end{theorem}
\begin{proof}
We denote the Hider and Searcher strategies described in the statement of the theorem by $h$ and $s$, respectively. We first show the Hider can guarantee the expected payoff of the game is at most $(p^{\floor*{(n+1)/2}} + p^{\ceil*{(n+1)/2}})/2$ by using $h$. Indeed, by symmetry, there are precisely two best responses for the Searcher to this strategy: $\sigma^C$ and $\sigma^A$. The expected payoff if the Searcher uses either of these strategies is
\[
P(h,\sigma^C)=P(h,\sigma^A) = \frac{1}{2} \prod_{i=1}^{\floor*{(n+1)/2}} p_i + \frac{1}{2} \prod_{i=1}^{\ceil*{(n+1)/2}} p_i = \frac{1}{2}(p^{\floor*{(n+1)/2}} + p^{\ceil*{(n+1)/2}}).
\]
Therefore, the value of the game is at most $(p^{\floor*{(n+1)/2}} + p^{\ceil*{(n+1)/2}})/2$. To prove that this is also a lower bound for the value, we consider the Searcher strategy $s$, and calculate the expected payoff when the Hider uses some pure strategy $j \in [n]$. By symmetry, we may assume that $j \le (n+1)/2$.
\begin{align*}
P(j, s) &= \frac{1}{2}\prod_{i=1}^j p_i + \frac{1}{2}\prod_{i=j}^n p_i \nonumber \\
& = \frac{1}{2}(p^j + p^{n+1-j}) \\
& = \frac{1}{2}p^j(1-p^{\ceil*{(n+1)/2}-j})(1-p^{\floor*{(n+1)/2}-j}) + \frac{1}{2}(p^{\floor*{(n+1)/2}} + p^{\ceil*{(n+1)/2}}) \\
&\ge \frac{1}{2}(p^{\floor*{(n+1)/2}} + p^{\ceil*{(n+1)/2}}).
\end{align*}
Therefore, the value of the game is at least $(p^{\floor*{(n+1)/2}} + p^{\ceil*{(n+1)/2}})/2$, and we must have equality.  Furthermore, strategies $h$ and $s$ are optimal.
\end{proof}

\subsection{Unequal detection probabilities}

We now consider the game in the case that the detection probabilities may not be equal. Note that for $n=2$, the vertices may be searched in any order, so the solution of the game is given by Theorem~\ref{thm:Lidbetter}. 

So we consider the game for $n \ge 3$, and we start by giving necessary and sufficient conditions that the Hider has an optimal strategy of a similar form to that of Theorem~\ref{thm:equal-p}.

For each vertex $j \in [n]$, we write $\pi_j$ for the product $p_1 p_2 \cdots p_j$ and we write $\bar{\pi}_j$ for the product $p_j p_{j+1} \cdots p_n$. Clearly, $\pi_j$ is decreasing in $j$ and $\bar{\pi}_j$ is increasing in $j$. Also, $\pi_1 > \bar{\pi}_1$ and $\pi_n < \bar{\pi}_n$. It follows that there exists a unique $j\in[n]$ such that $\pi_j \ge \bar{\pi}_j$ and $\pi_{j+1} < \bar{\pi}_{j+1}$. 

\begin{lemma} \label{lem:j=1}
Let $j$ be such that $\pi_j \ge \bar{\pi}_j$ and $\pi_{j+1} < \bar{\pi}_{j+1}$. Suppose the following condition holds.
\begin{align}
\frac{\pi_i - \pi_j}{\bar{\pi}_j - \bar{\pi}_i} &\ge \frac{\pi_j - \pi_{j+1}}{\bar{\pi}_{j+1}-\bar{\pi}_j} \text{ for all $i \neq j, j+1$}; \label{eq:cond1} 
\end{align}
Then the value of the game $\Gamma(C_n)$ is given by
\begin{align}
V \equiv \frac{\bar{\pi}_{j+1}\pi_j-\pi_{j+1}\bar{\pi}_{j}}{\pi_j-\pi_{j+1}+\bar{\pi}_{j+1}-\bar{\pi}_j}. \label{eq:noneqV}
\end{align}
It is optimal for the Hider to choose vertex $j$ with probability $q$ and vertex $j+1$ with probability $1-q$, where
\begin{align}
q=\frac{\bar{\pi}_{j+1}-\pi_{j+1}}{\pi_j-\pi_{j+1}+\bar{\pi}_{j+1}-\bar{\pi}_j}. \label{eq:noneq-hider}
\end{align}
It is optimal for the Searcher to choose strategy $\sigma^C$ with probability $r$ and strategy $\sigma^A$ with probability $1-r$, where
\begin{align}
r = \frac{\bar{\pi}_{j+1} - \bar{\pi}_{j}}{\pi_j-\pi_{j+1}+\bar{\pi}_{j+1}-\bar{\pi}_j}. \label{eq:noneq-searcher}
\end{align}
\end{lemma}
\begin{proof}
First consider a restricted version of the game where the Hider's pure strategy set is reduced to only vertices $j$ and $j+1$. The value of this game is at most the value of the original game and all Searcher strategies are weakly dominated by the strategies $\sigma^C$ and $\sigma^A$. This is because for any Searcher strategy, if vertex $j$ is searched before vertex $j+1$, then the Searcher must search all of vertices $0,1,\ldots,j-1$ before searching vertex $j$, and of all possible strategies which have this property, $\sigma^C$ clearly maximizes the payoff against the Hider strategies $H=j$ and $H=j+1$. Similarly if $j+1$ is searched before $j$.

It is easy to verify that for this $2 \times 2$ game, the value is $V$, as defined in~(\ref{eq:noneqV}) and optimal strategies are given by~(\ref{eq:noneq-hider}) and~(\ref{eq:noneq-searcher}).

So to complete the proof, we just need to check that the Searcher strategy given by~(\ref{eq:noneq-searcher}) guarantees a payoff of at least $V$ for any pure strategy $i \neq j,j+1$ of the Hider.  Indeed, in this case, if the Searcher uses the strategy given by~(\ref{eq:noneq-searcher}), the expected payoff is 
\[
\frac{(\bar{\pi}_{j+1} - \bar{\pi}_{j})\pi_i + (\pi_j - \pi_{j+1})\bar{\pi}_i}{\pi_j-\pi_{j+1}+\bar{\pi}_{j+1}-\bar{\pi}_j}.
\]
This expected payoff is at least $V$ if and only if Condition~(\ref{eq:cond1}) holds.
\end{proof}

In the special case that all the $p_i$'s are equal to some~$p$, it is easy to verify that Condition~(\ref{eq:cond1}) reduces to $p^{j-i} \le 1$ for $i <j$ and it reduces to $p^{i-j} \le 1$ for $i > j$, both of which are trivially true.  Therefore, Lemma~\ref{lem:j=1} gives an alternative proof of Theorem~\ref{thm:equal-p}. 

While the conditions of Lemma~\ref{lem:j=1} seem rather abstract, in the examples we have considered, they are usually satisfied. 

Consider the case $n=3$. Then Condition~(\ref{eq:cond1}) reduces to $(1-p_2)^2 \ge p_2(1-p_1)(1-p_3)$.
We checked whether this condition holds for the 729 possible choices of the parameters $(p_1,p_2,p_3)$, given by choosing one of the values $1/10,2/10,\ldots,9/10$ for each $p_i$. Out of these 729 games, the condition was met in 526, or 72\% of cases.

We finish this section by giving a full solution to the game for $n=3$. Without loss of generality, we assume that $p_1 \ge p_3$ (otherwise we could relabel the vertices in reverse order).

\begin{proposition}
Consider the game $\Gamma(C_3)$, where $p_1 \ge p_3$. The solution of the game splits into two cases as follows.

\textbf{Case 1.} If $(1-p_2)^2 \ge p_2(1-p_1)(1-p_3)$ then an optimal strategy for the Hider is to hide at vertices 2 and 3 with probabilities proportional to $p_3(1-p_1p_2)$ and $p_2(p_1-p_3)$, respectively. An optimal strategy for the Searcher is to choose $\sigma^C$ and $\sigma^A$ with probabilities proportional to $p_3(1-p_2)$ and $p_1p_2(1-p_3)$, respectively. The value of the game is 
\[
\frac{p_1 p_2 p_3(1-p_2p_3)}{p_1p_2(1-p_3)+p_3(1-p_2)}. \]

\textbf{Case 2.} If $(1-p_2)^2 < p_2(1-p_1)(1-p_3)$ then it is optimal for the Hider to hide at vertices $i$ with probability proportional to $(1-p_i)/p_i$ for $i=1,2,3$. It is optimal for the Searcher to choose $\sigma^C$ and $\sigma^A$ with probabilities $q$ and $r$, where
\[
q = \frac{p_3(1-p_2)(1-p_1 p_2 p_3)}{(1-p_2 p_3)(p_1 p_2 (1-p_3) + p_1 p_3(1-p_2) + p_2 p_3(1-p_1))}
\]
and
\[
r = \frac{p_1(p_2+p_3)+p_2p_3(3+p_2p_3)}{(1-p_2 p_3)(p_1 p_2 (1-p_3) + p_1 p_3(1-p_2) + p_2 p_3(1-p_1))}.
\]
With probability $1-q-r$ the Searcher searches the vertices in the order $(1,3,2)$.

The value of the game is
\begin{align}
\frac{1-p_1 p_2 p_3}{(1-p_3)/p_3 +  (1-p_2)/p_2 +  (1-p_1)/p_1}. \label{eq:n=3}
\end{align}
\end{proposition}

\begin{proof}
For Case 1, we note that $p_1 \ge p_2$ implies that $\pi_2 \ge \bar{\pi}_2$ and $\pi_3 < \bar{\pi}_3$. So, taking $j=2$ and noting that Condition~(\ref{eq:cond1}) is equivalent to the condition $(1-p_2)^2 \ge p_2(1-p_1)(1-p_3)$, the solution of the game is given by Lemma~\ref{lem:j=1}. The optimal strategies reduce to those given in the statement of this Proposition.

For Case 2, Lemma~\ref{lem:upperbound} shows that the Hider strategy given in the statement of the Proposition ensures a payoff of at least that given in~(\ref{eq:n=3}). To verify that this payoff is also achieved by the Searcher strategy described in the statement of the Proposition is a straightforward algebraic exercise which we leave to the reader. However, we must also check that the probabilities $q$, $r$ and $1-q-r$ are indeed probabilities: in particular that they are non-negative.  It is clear that $q$ and $r$ are non-negative, and to show that $1-q-r$ is non-negative we compute
\begin{align*}
    (1-p_2 p_3)(p_1 p_2 (1-p_3) &+ p_1 p_3(1-p_2) + p_2 p_3(1-p_1))(1-q-r)\\ 
    &= 2p_2p_3+p1p_2^2p_3^2-p_1p_2^2p_3-p_2^2p_3^2-p_3 \\
    & \ge (p_1p_2^2+p_2^2p_3-p_1p_2^2p_3+1)p_3+p1p_2^2p_3^2-p_1p_2^2p_3-p_2^2p_3^2-p_3 \\
    &= 0,
\end{align*}
where the inequality follows from $(1-p_2)^2 \ge p_2(1-p_1)(1-p_3)$.
\end{proof}

\subsection{An approximately optimal Searcher strategy}

In this section we present a Searcher strategy that is approximately optimal when the probabilities $p_i$ are not too ``spread out''. To do this, we start by considering a related game which is similar to the one we have studied thus far, but with a slightly different payoff function. In particular, the payoff $P'(j,\sigma)$ for given Hider and Searcher strategies $j$ and $\sigma$ is given by
\[
P'(j,\sigma) \equiv \frac{P(j,\sigma)}{\sqrt{p_j}} \equiv \sqrt{p_j} \prod_{\{i: i < \sigma^{-1}(j)\}} p_{\sigma(i)}.
\]
We denote this new game played on a graph $G$ by $\Gamma'(G)$. 

Similarly to the game $\Gamma(C_n)$, for each vertex $j \in[n]$, we write $\pi'_j$ for the product $p_1  \cdots p_{j-1} \sqrt{p_j}$ and we write $\bar{\pi}'_j$ for the product $\sqrt{p_j} p_{j+1} \cdots p_n$. Note that $\pi'_j \bar{\pi}'_j = \pi_n$ for any $j$. As before, $\pi'_j$ is decreasing in $j$ and $\bar{\pi}'_j$ is increasing in $j$. Also, $\pi'_1 > \bar{\pi}'_1$ and $\pi'_n < \bar{\pi}'_n$. So there exists a unique $j\in[n]$ such that $\pi'_j \ge \bar{\pi}'_j$ and $\pi'_{j+1} < \bar{\pi}'_{j+1}$.

\begin{lemma}
Let $j$ be such that $\pi'_j \ge \bar{\pi}'_j$ and $\pi'_{j+1} < \bar{\pi}'_{j+1}$. Then there is an optimal Searcher strategy for $\Gamma'(C_n)$ that chooses every strategy except possibly $\sigma^C$ and $\sigma^A$ with probability 0. There is an optimal Hider strategy that chooses every strategy except $j$ and $j+1$ with probability 0.
\end{lemma}
\begin{proof}
First suppose  $\pi'_j = \bar{\pi}'_j = \sqrt{\pi_n}$. In this case, consider the Hider strategy that chooses $j$ with probability 1. This guarantees a payoff of at most $\pi'_j=\bar{\pi}'_j$. 

Consider the Searcher strategy $s$ that chooses each of $\sigma^C$ and $\sigma^A$ with probability $1/2$. For $\pi_n < x < 1$, let $f(x)= x/2 + \pi_n/(2x)$, which has a minimum at $x=\pi'_j$. Then for any Hider strategy $i \neq j$, the expected payoff when the Searcher plays $s$ is
\[
P'(s,i)=(1/2) \pi'_i + (1/2)\bar{\pi}'_i = f(\pi'_i) \ge f(\pi'_j) = P'(s,j).
\]
So the strategy $s$ guarantees a payoff of at least $P'(s,j)=\pi'_j$.

Now suppose $\pi'_j > \bar{\pi}'_j$. Consider the $2\times 2$ game whose payoffs are given by $P'$, where the Searcher's strategy set is $\{\sigma^C,\sigma^A\}$ and the Hider's strategy set is $\{j,j+1\}$. Then it is easy to verify that both players have unique optimal strategies in this game where they play each of their pure strategies with positive probability. 

Let $s$ denote the optimal Searcher strategy and $h$ the optimal Hider strategy in this $2 \times 2$ game. Since each of the players' pure strategies must be best responses to the optimal strategy of the other player, the value $v'$ of this game is given by
\[
v' = P'(s,j)= P'(s,j+1).
\]
Clearly, the Hider strategy $h$ also guarantees an expected payoff of at most $v$ in the game $\Gamma'$, since $\sigma^C$ and $\sigma^A$ are the only best responses to $h$ in $\Gamma'$. We will show that the Searcher strategy $s$ guarantees an expected payoff of at least $v$ in $\Gamma'$.

Let $\beta$ be the probability that the Searcher uses the strategy $\sigma^C$ in $s$. For $\pi_n < x < 1$, let $g(x) = \beta x + (1-\beta) \pi_n/x$. Note that for any Hider strategy $i$, the payoff under $s$ is $P'(i,s) = g(\pi'_i)$.  Since $g$ is a convex function of $x$ and $g(\pi'_j)=g(\pi'_{j+1})=v$, it must be the case that $g(\pi'_i) \ge g(\pi'_j)=v$ for all $i \in [n]$. It follows that $s$ is also optimal in $\Gamma'$. 
\end{proof}

We now show how the solution of the game $\Gamma'$ can be exploited to give approximately optimal solutions to $\Gamma$.
\begin{proposition} \label{prop:approx}
Let $j$ be such that $\pi'_j \ge \bar{\pi}'_j$ and $\pi'_{j+1} < \bar{\pi}'_{j+1}$. Let $\alpha = \max \{\sqrt{p_j},\sqrt{p_{j+1}}\}/\min_i \sqrt{p_i}$ and let $v$ be the value of the game $\Gamma(C_n)$. Then any optimal Searcher  strategy for $\Gamma'(C_n)$ guarantees an expected payoff of at least $v/\alpha$ in $\Gamma(C_n)$ and any optimal Hider strategy for $\Gamma'(C_n)$ guarantees an expected payoff of at most $\alpha v$ in $\Gamma(C_n)$.
\end{proposition} 
\begin{proof}
Let $s$ and $h$ be optimal Searcher and Hider strategies in $\Gamma'$. Since $P'(i,\sigma) = P(i,\sigma)/\sqrt{p_i}$ for any $\sigma$ and $i$, the Searcher can ensure a payoff of at least $\min_i \sqrt{p_i}v'$ in $\Gamma$ by using $s$. The Hider can ensure a payoff of at most $\max\{\sqrt{p_j},\sqrt{p_{j+1}}\} v'$ by using $h$, because the support of $h$ is contained in $\{j,j+1\}$. Therefore,
\[
\min_i \sqrt{p_i}v' \le v \le \max\{\sqrt{p_j},\sqrt{p_{j+1}}\} v'.
\]
It follows that the strategy $s$ ensures the payoff is at least 
\[
\min_i \sqrt{p_i}v' \ge \min_i \sqrt{p_i} \cdot \frac{v}{\max\{\sqrt{p_j},\sqrt{p_{j+1}}\}} = \frac{v}{\alpha}.
\]
Similarly, the Hider strategy $h$ ensures the payoff is at most
\[
\max\{\sqrt{p_j},\sqrt{p_{j+1}}\} v' \le \max\{\sqrt{p_j},\sqrt{p_{j+1}}\} \cdot \frac{v}{\min_i \sqrt{p_i}} = \alpha v
\]
\end{proof}

Note that Proposition~\ref{prop:approx} provides an alternative proof of Theorem~\ref{thm:equal-p}. Moreover, if $p_j=p_{j+1} \le p_i$ for all $i \in [n]$, then Proposition~\ref{prop:approx} gives optimal strategies for $\Gamma$. Finally, since $\max\{\sqrt{p_j},\sqrt{p_{j+1}}\} \le 1$, it is always true that $\alpha \le 1/\min_j \sqrt{p_j}$, so if all the probabilities are at least $\beta$, for some $\beta$, then an optimal Searcher strategy for $\Gamma'$ ensures a payoff of at least $\beta v$ in~$\Gamma'$.

\section{Multiple targets with equal detection probabilities}
\label{sec:multiple}

We now consider the search and rescue game in the case that the number $k$ of targets is greater than~1. Of course this complicates the game, but in Subsection~\ref{sec:non-adapt}, we explain how the the solution for the case $k=n-1$ follows easily from the previous work of \cite{Lidbetter20}.

We then consider a variation on the game in Subsection~\ref{sec:adapt} in which we enlarge the Searcher's strategy set to allow her to use {\em adaptive} search strategies, where at any time she can choose which vertex to search next based on information gathered up to that time.  To distinguish this variation of the game from the original, we call the original {\em non-adaptive search}. We focus on the case of equal detection probabilties.

\subsection{Non-adaptive search} \label{sec:non-adapt}

The solution of the non-adaptive game appears elusive in general, but in the special case of $k=n-1$, it follows from previous work. Indeed, for this case, a Searcher strategy is completely specified by the final vertex to be searched, since the order of search of the first $k$ vertices does not matter. It follows that the network structure of the problem provides no restriction to the Searcher's strategy set, so that the solution of the game on a cycle follows from the solution of the game with no network structure, as given in \cite{Lidbetter20}.  We summarize this observation below.

\begin{proposition}
    The solution of the search and rescue game with $k=n-1$ targets played on $C_n$ is given by Theorem~\ref{thm:Lidbetter}.
\end{proposition}

\subsection{Adaptive search} \label{sec:adapt}

We now turn to the adaptive case, assuming that all the detection probabilities are equal to some~$p$. We represent a Hider strategy by a $k$-tuple $(v_1,\ldots,v_k)$ such that $1 \le v_1<v_2 < \ldots < v_k \le n$, so that $v_j$ is the location of the $j$th target. Let $\mathcal{H}_k$ be the set of all such $k$-tuples. Let $v_0=0$ and $v_{k+1}=n+1$, and for $n \ge k$ let 
\[
S_{n,k} = \{(v_1,\dots,v_k) \in \mathcal{H}_k: v_{i+1}-v_i = \floor{(n+1)/(k+1)} \text{ or } \ceil{(n+1)/(k+1)}, i=0,\ldots,k\}. 
\]
Also let 
\[
S_{n,k}^- = \{(v_1,\ldots,v_k) \in S_{n,k}: v_1=\floor{(n+1)/(k+1)} \}
\]
and
\[
S_{n,k}^+ = \{(v_1,\ldots,v_k) \in S_{n,k}: v_1=\ceil{(n+1)/(k+1)} \}.
\]
Note that $S_{n,k}^-$ and $S_{n,k}^+$ partition $S_{n,k}$ unless $n+1$ is divisible by $k+1$, in which case 
\[
S_{n,k}^-=S_{n,k}^+=S_{n,k} = \{((n+1)/(k+1),2(n+1)/(k+1),\ldots , k(n+1)/(k+1)) \}.
\]
Let $s_{n,k}=|S_{n,k}|$ be the cardinality of $S_{n,k}$; also let $s_{n,k}^-=|S_{n,k}^-|$ and $s_{n,k}^+ = |S_{n,k}^+|$. We have already seen that if $n+1$ is divisible by $k+1$ then $s_{n,k} = 1$ since $S_{n,k}$ is a singleton. If $n+1$ is not divisible by $k+1$ then we can write $n=a(k+1)+b$ where $a$ and $b$ are non-negative integers and $b \le k$. In this case $\floor{(n+1)/(k+1)}=a$ and $\ceil{(n+1)/(k+1)} = a+1$. The following relation is immediate from the definition of $S_{n,k}$.
\begin{align}
s_{n,k}^- = s_{n-a,k-1} \text{ and } s_{n,k}^+ = s_{n-a-1,k-1}. \label{eq:recursion}
\end{align}
\begin{lemma} \label{lem:counting}
Suppose $n=a(k+1)+b$, where $a$ and $b$ are positive integers. If $b \le k$, then
\[
s_{n,k}={k+1 \choose b+1}.
\]
If $b\le k-1$, then 
\[
s_{n,k}^-= {k \choose b+1} \text{ and } s_{n,k}^+ = {k \choose b}.
\]
\end{lemma}
\begin{proof}
We have already shown that $s_{n,k}=1={k+1 \choose b+1}$ for $b=k$ so we restrict our attention to the case $b \le k-1$, proving the lemma by induction on $k$. If $k=1$ and $b=0$ then $n$ is even, and $S_{n,k} = \{(n/2),(n/2+1)\}$. The expressions for $s_{n,k},s_{n,k}^-$ and $s_{n,k}^+$ are easy to verify.

Now suppose $k \ge 2$ and that the lemma is true for all smaller values of $k$. Then $n-a=ak+b$ where $0 \le b \le k-1$, so by~(\ref{eq:recursion}) and the induction hypothesis,
\[
s_{n,k}^- = s_{n-a,k-1} = {k \choose b+1}.
\]
Similarly, $n-a-1 = ak + (b-1)$ where $-1 \le b-1 \le k-2$. If $b \ge 1$, then
\[
s_{n,k}^+ = s_{n-a-1,k-1} = {k \choose (b-1)+1} = {k \choose b}.
\]
If $b=0$, then we write $n-a-1 = (a-1)k+ (k-1)$, so that
\[
s_{n,k}^+ = s_{n-a-1,k-1} = {k \choose (k-1)+1} = 1 = {k \choose b}.
\]
Finally, we note that
\[
s_{n,k} = s_{n,k}^- + s_{n,k}^+ = {k \choose b+1}+{k \choose b} = {k+1 \choose b+1}.
\]
\end{proof}

We can now describe the optimal Hider strategy. The Hider simply chooses from each of the strategies in $S_{n,k}$ with equal probability. An important property of this strategy is that after the Searcher finds the first target, the remaining targets are hidden optimally among the unsearched vertices. 

There are only $k+1$ (weakly) undominated Searcher strategies, which we denote $\sigma_j, ~j=0,1,\ldots,k$. Strategy $\sigma_j$ searches the vertices in a clockwise direction until finding $j$ targets, then searches the vertices in an anticlockwise direction, starting from the root. Note that $\sigma_k$ is equivalent to $\sigma^C$ and $\sigma_0$ is equivalent to $\sigma^A$. To illustrate why any other Searcher strategy would be weakly dominated, first note that it is evident that any Searcher strategy that is not weakly dominated can be specified by two sequences $x_1,\ldots,x_t$ and $y_1,\ldots,y_t$ taking values in $0,1\ldots,k$ such that $\sum_{i=1}^t x_i+ y_i = k$, the interpretation being that Searcher goes clockwise until finding $x_1$ targets, then anticlockwise until finding $y_1$ targets, then clockwise until finding $x_2$ targets, and so on. But such a strategy has the same payoff against any Hider strategy as the strategy $\sigma_j$, where $j=\sum_{i=1}^t x_i$.

\begin{theorem}
Let $n=a(k+1)+b$ where $a$ and $b$ are non-negative integers and $b \le k$. The value of the game =$\Gamma(C_n)$ with an adaptive Searcher, equal detection probabilities and $k$ targets is
\begin{align}
\left( \frac{k-b}{k+1} \right) p ^{n-a+1} + \left( \frac{b+1}{k+1} \right) p ^{n-a} \label{eq:k-value}
\end{align}
An optimal strategy for the Searcher is to choose equiprobably between the strategies $\sigma_0,\sigma_1,\ldots,\sigma_k$. An optimal strategy for the Hider is to choose equiprobably between the strategies in $S_{n,k}$.
\end{theorem}
\begin{proof}
From the construction of the Hider strategy, it is clear that any Searcher strategy $\sigma_j$ will win with the same probability. So we calculate the expected payoff of $\sigma_0=\sigma^A$. The payoff depends on whether $v_1$ is equal to $\floor{(n+1)/(k+1)}=a$ or $\ceil{(n+1)/(k+1)}=a+1$. By Lemma~\ref{lem:counting}, the probability that $v_1$ is equal to $a$ is $(k-b)/(k+1)$ and the probability $v_1$ is equal to $a+1$ is $(b+1)/(k+1)$. It follows that the expected payoff is given by Equation~(\ref{eq:k-value}).

Now consider a fixed Hider strategy $(v_1,\ldots,v_k)$ that is a best response to the Searcher strategy $s$ described in the statement of the theorem. Let $n_j=v_{j +1}- v_{j}$ for $j=0,1,\ldots,k$, where $v_0=0$ and $v_{k+1}=n+1$. The expected payoff of this Hider strategy against $s$ is
\begin{align}
\frac{1}{k+1}\sum_{j=0}^{k} p^{n-n_j+1}. \label{eq:exppayoff}
\end{align}
We claim that $|n_i-n_j| \le 1$ for all $i,j$. Suppose this is not true, and that $n_i \ge n_j+2$ for some $i,j$. Without loss of generality, suppose that $i < j$. Then consider the Hider strategy obtained by moving targets $(i+1)$ through $(j+1)$ one vertex anticlockwise. In this case, the values of $n_0,\ldots,n_k$ all stay the same except for $n_i$ which decreases by 1 and $n_j$ which increases by 1.  Therefore, the difference between the expected payoffs of the original strategy and the new strategy is
\[
\frac{1}{k+1}(p^{n-n_i+1} + p^{n-n_j+1} - p^{n-n_i+2} - p^{n-n_j}) = \frac{p^{n-n_i+1}}{k+1}(1 + p^{n_i-n_j-1})(1-p) > 0.
\]
This contradicts $(v_1,\ldots,v_k)$ being a best response to $s$.  So we have established that $|n_i-n_j| \le 1$ for all $i,j$. Since $\sum_{j=0}^k n_j = n+1=a(k+1)+b+1$ and $1 \le b+1 \le k+1$, it must be the case that $k-b$ of the parameters $n_j$ are equal to $a$ and $b+1$ are equal to $a+1$. Therefore, by~(\ref{eq:exppayoff}), the expected payoff is
\[
\frac{1}{k+1}((k-b)p^{n-a+1} + (b+1)p^{n-(a+1)+1}) = \left( \frac{k-b}{k+1} \right) p ^{n-a+1} + \left( \frac{b+1}{k+1} \right) p ^{n-a}.
\]
\end{proof}

\section{A continuous version of the game}
\label{sec:cont}

Since the non-adaptive variant of the game with $k \ge 2$ seems to be difficult to analyze, we consider a continuous version of the game, in the hope that this might offer further insight. The game is played on a unit length circle $C$, with root $O$. We consider $C$ as the unit interval $[0,1]$ with the points $0$ and $1$ identified, and equipped with Lebesgue measure $\lambda$, so that $\lambda(A)$ is the measure (or length) of a measurable subset $A$ of $[0,1]$. When referring to a point on $C$, we measure the distance clockwise around the circle from $O$.

We use the expanding search paradigm of \cite{AL13} for continuous rooted networks, defined as follows. An expanding search is a family of connected sets $S(t) \subseteq C$ for $t \in [0,1]$ satisfying
\begin{enumerate}
\item $S(0) = \{O\}$,
\item $S(t) \subset S(t')$ for $t < t'$,
\item $\lambda(S(t)) = t$ for $t \in [0,1]$.
\end{enumerate}
The set $S(t)$ corresponds to the part of the network searched by time $t$.

We consider a game where the Hider chooses $k$ points on the circle: that is a subset $H$ of $C$ of cardinality $k$.  For strategies $S$ and $H$, let $T(S,H)=\inf \{t: H \subset S(t) \}$ be the first time the search contains all $k$ points. We call this the {\em search time}. We extend the definition of $T$ so that $T(s,h)$ denote the expected search time for mixed strategies $s$ and $h$.

In the discrete game with constant detection probabilities $p$ considered in the previous section of this paper, the probability that the Searcher does not get captured after searching $j$ vertices is $p^j$. Mirroring this assumption in the continuous case, we assume that the probability the Searcher does not get captured by time $t$ is $p^t$, where $p \in (0,1)$. So the payoff of the game, which is the probability all targets are found, is given by $p^{T(S,H)}$. As before, the Hider is the minimizer and the Searcher is the maximizer.

\cite{Lidbetter13} considered a game with the same strategy sets, but where the payoff was $T(S,H)$.  Also, the Searcher was the minimizer and the Hider was the maximizer.  We will consider the performance of the strategies proven to be optimal in the game \cite{Lidbetter13}, starting with the case of adaptive search, followed by that of non-adaptive search.

Theorem~\ref{thm:cont-adapt} gives a solution to the continuous game in the adaptive setting. It is worth pointing out that the optimal strategies are very similar to those in the discrete game and are precisely the same as the optimal strategies in the adaptive version of the continuous game considered in~\cite{Lidbetter13}.

\begin{theorem} \label{thm:cont-adapt}
In the continuous search and rescue game with $k$ targets in the adaptive setting, it is optimal for the Hider to use the pure strategy $H=\{1/(k+1),2/(k+1),\ldots,k/(k+1)\}$. It is optimal for the Searcher to pick each of the following pure strategies with probability $1/(k+1)$: search the circle in the clockwise direction until $j$ objects are found, then search the circle in the anti-clockwise direction (for each $j=0,1,\ldots,k$). The value of the game is $p^{k/(k+1)}$.
\end{theorem}

\begin{proof}
First note that against the Hider strategy $h$ given in the statement of the theorem, the expected search time $T(S,h)$ cannot be less than $k/(k+1)$ for any Searcher strategy $S$. Hence, the expected payoff must be at most $p^{k/(k+1)}$.

Now suppose the Searcher uses the strategy given in the statement of theorem and the Hider uses some arbitrary pure strategy $H=\{x_1,\ldots,x_k\}$, where $0 < x_1 < \cdots <x_k < 1$. Then setting $x_0=0$ and $x_{k+1}=1$, the search time is $1-(x_{j+1}-x_j)$ with probability $1/{k+1}$ for each $j=1,\ldots,k+1$. Hence the expected payoff is
\[
\sum_{j=1}^{k+1} \frac{1}{k+1} p^{1-(x_{j+1}-x_j)}.
\]
Since the function $p^t$ is convex, it follows from Jensen's inequality and the monotonicity of $p^t$ that the expected payoff is at least
\[
p^{ \frac{1}{k+1} \sum_{j=1}^{k+1} 1- (x_{j+1}-x_j)} = p^{k/(k+1)}.
\]
\end{proof}

Note that for any game with the same strategy sets as those of the continuous search and rescue game in the adaptive setting, if the payoff is some convex function of the search time then the strategies of Theorem~\ref{thm:cont-adapt} will be optimal.

We now turn to the continuous game in the non-adaptive setting, first describing the Hider strategy used in \cite{Lidbetter13}, which we will call $h^*$. First the Hider picks a number $x$ uniformly at random from the interval $[0,1/k]$ then hides the targets at the points $\{x,x+1/k,x+2/k,\ldots,x+(k-1)/k\}$. An important property of this strategy is for any Searcher strategy, after time $1-1/k$, there will be one undiscovered target remaining, hidden uniformly at random in the unsearched region.

Next we describe the Searcher strategy, which we will denote by $s^*$. The Searcher picks an integer $j$ uniformly at random between $1$ and $k$ then picks with equal probability ``clockwise'' or ``anticlockwise''.  The Searcher then travels in the chosen direction for distance $j/k$ before traveling from $O$ in the other direction for distance $(k-j)/k$. Theorem 3.8 of \cite{Lidbetter13} showed that the expected search time of this strategy against any Hider strategy is at most $1-1/(2k)$.

\begin{proposition} \label{prop:cont-nonadapt}
The value $V$ of the game satisfies
\begin{align}
p^{1-1/(2k)} \le V \le  \frac{k(p- p^{1-1/k})}{\log p}. \label{eq:bounds}
\end{align}
\end{proposition}
\begin{proof} For the lower bound, we consider the Searcher strategy $s^*$ described above. As remarked, the expected search time $T(s^*,H)$ against any Hider strategy $H$ satisfies $T(s^*,H) \le 1-1/(2k)$. Again, applying Jensen's inequality, the expected payoff of $s^*$ against $H$ is at least $p^{1-1/(2k)}$.

For the upper bound, we consider the Hider strategy $h^*$ described above, and we use the fact that against this strategy, any Searcher strategy will have found $k-1$ objects by time $1-1/k$, and the final object will be uniformly hidden in the remaining part of the circle that has not been searched. Therefore the expected payoff is
\[
\int_{1-1/k}^1 p^t \cdot k~dt = \left[ \frac{ k p^t}{\log p} \right]_{t=1-1/k}^1 = \frac{ k(p- p^{1-1/k})}{\log p}.
\]
\end{proof}

Let $f(p)$ be the ratio of the upper to lower bounds, so that 
\[
f(p) = \frac{k(p-p^{1-1/k})}{p^{1-1/(2k)}\log p}=\frac{k(p^{1/(2k)}-p^{-1/(2k)})}{\log p}.
\]

Then the first derivative of $f(p)$ is
\begin{align*}
f'(p)%& =\frac{k[(1/{2k})p^{1/(2k)-1}-(-1/{2k})p^{-1/(2k)-1}]\log p-k/p(p^{1/(2k)}-p^{-1/(2k)})}{\log^2p}\\
& =\frac{\log p(p^{\frac{1}{2k}}+p^{-\frac{1}{2k}})-k(p^{\frac{1}{2k}}-p^{-\frac{1}{2k}})}{2p\log^2p}
\end{align*}
Let $$g(p)=\log p(p^{\frac{1}{2k}}+p^{-\frac{1}{2k}})-k(p^{\frac{1}{2k}}-p^{-\frac{1}{2k}})$$
be the denominator of $f'(p)$.
Then the first derivative of $g(p)$ is
\[
g'(p)=\frac{1}{2k}\log p(p^{\frac{1}{2k}-1}-p^{-\frac{1}{2k}-1}).
\]
It is easy to see that $g'(p)$ is positive for all $p \in (0,1)$, since $\log p < 0$ and $p^{1/(2k)-1} < p^{-1/(2k)-1}$. Therefore, $g(p)<g(1)=0$.

Since the denominator of $f'(p)$ is positive, $f'(p)$ must be negative. It follows that $f(p)$ is monotonically decreasing. Moreover, we can calculate the limit of the ratio $f(p)$ as $p$ approaches 1 by applying L'H\^opital's rule as follows.
\[
\lim_{p \rightarrow 1} f(p) = \lim_{p \rightarrow 1} \frac{k(1/(2k)p^{1/(2k)-1} + 1/(2k) p^{-1/(2k)-1})}{1/p}= 1.
\]
Similarly, we can calculate the limit as $p$ approaches $0$:
\begin{align*}
\lim_{p \rightarrow 0} f(p) &= \lim_{p \rightarrow 0} \frac{k(1/(2k)p^{1/(2k)-1} + 1/(2k) p^{-1/(2k)-1})}{1/p} \\
&= \lim_{p \rightarrow 0} k(1/(2k)p^{1/(2k)} + 1/(2k) p^{-1/(2k)}) \\
&=\infty.
\end{align*}
So to sum up, the ratio of the bounds decreases as $p$ increases, and is asymptotically equal to 1 when $p$ approaches 1. In other words, the two strategies of Proposition~\label{prop:cont-nonadapt} are asymptotically optimal as the success probability tends to 1. 

Although the ratio is asymptotically equal to $\infty$ as $p$ approaches 0, the {\em absolute} difference between the upper and lower bounds is small. As shown in Figure~\ref{fig:graph}, for $k=2$ the two bounds are very close to each other in absolute terms.  In fact, computational methods show that the maximum difference $\max_{p \in [0,1]} k(p-p^{1-1/k})/\log(p)-p^{1-1/(2k)}$ between the upper and lower bounds for $k=2$ is approximately $0.0103$, occurring at roughly $x=0.0653$.  Furthermore, the maximum difference between the bounds for $k$ up to $9999$ is never any greater than $0.0103$.

\begin{figure}[h!]
\centering
\includegraphics[scale=0.1]{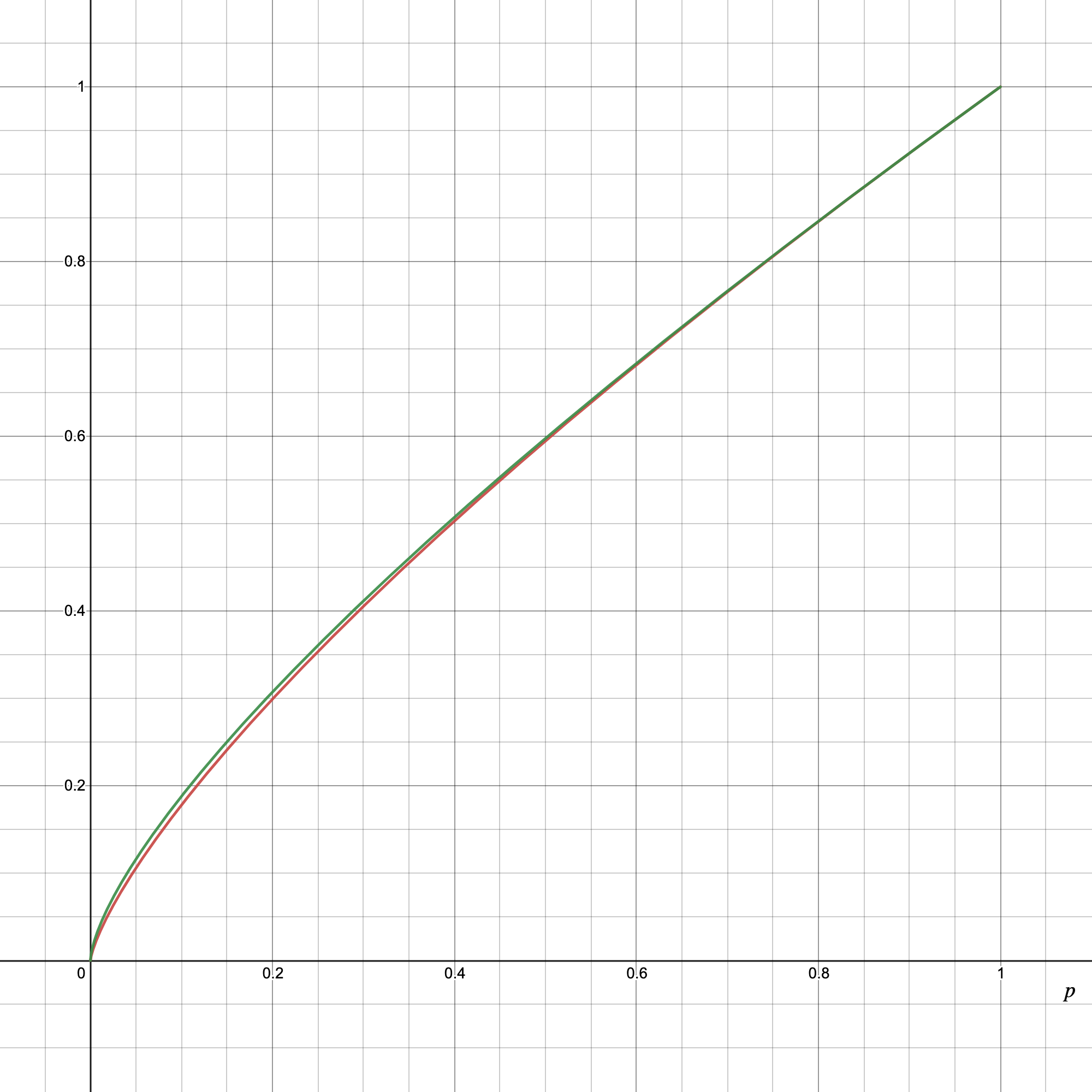}
\caption{Upper and lower bounds for the value of the game when $k=2$.}
\label{fig:graph}
\end{figure}

%Suppose $p$ is close to $1$, so that $p = 1-\varepsilon$, for some small $\varepsilon >0$. Then the ratio of the upper to lower bounds in~(\ref{eq:bounds}) is
%\begin{align*}
 %\frac{k(p- p^{1-1/k})}{p^{1-1/(2k)} \log p } &= \frac{k((1-\varepsilon)^{1/(2k)} - (1-\varepsilon^{-1/(2k)}))}{\log (1- \varepsilon)} \\
 %& = \frac{k((1 - \varepsilon/(2k) + \mathcal O(\varepsilon^2)) - (1 + \varepsilon/(2k)+ \mathcal O(\varepsilon^2)))}{-\varepsilon + \mathcal O(\varepsilon^2)} \\
% &= \frac{\varepsilon+ \mathcal O(\varepsilon^2)}{\varepsilon+ \mathcal O(\varepsilon^2)} \\
 %&= 1 + \mathcal{O}(\varepsilon).
%\end{align*}
%So for $p$ close to 1, the two bounds are close to each other.  

%This could be made more precise by proving that the ratio of the bounds is decreasing in $p$. How big does $p$ have to be so that the ratio is less than, say, 2?

%Questions: For small $p$ is the value closer to the upper bound in~\ref{eq:bounds} or is it closer to the lower bound?  Perhaps some experiments with the discrete version of the game could give a clue.  If we cannot solve the game precisely (which would be nice, even for $k=2$), what are some better strategies for small $p$?

\section{Conclusion}
\label{sec:conc}

We have analyzed the search and rescue game on the simplest networks that are not trees: cycles. Even for a single hidden target, the game is not trivial to solve unless all the detection probabilities are equal. Indeed, for equal detection probabilities, the adaptive version of the game admits a neat solution for an arbitrary, known number of hidden targets. The non-adaptive version of the the game with multiple targets is much harder to analyze. Finding approximate solutions to a continuous version of this game may offer some clues as to how to find approximate solutions to the discrete game in further work.

\section*{Acknowledgements} This material is based upon work supported by the National Science Foundation under Grant No. IIS-1909446.

\end{document}